\documentclass[a4paper,twoside]{amsart}
\usepackage{amsmath}
\usepackage{amsfonts}
\usepackage{amssymb}
\usepackage{amsthm}
\usepackage{newlfont}
\usepackage{graphicx}
\usepackage{amscd}

\theoremstyle{plain}
\newtheorem{Th}{Theorem}[section]
\newtheorem{Cor}[Th]{Corollary}
\newtheorem{Lem}[Th]{Lemma}
\newtheorem{Prop}[Th]{Proposition}
\theoremstyle{definition}

\theoremstyle{remark}
\newtheorem*{Rem}{Remark}
\numberwithin{equation}{section}
\newcommand{\PP}{{\mathbb P}}

\newcommand{\DD}{{\mathbb D}}

\newcommand{\ZZ}{{\mathbb Z}}

\newcommand{\bphi}{\boldsymbol{\phi}}
\newcommand{\bPhi}{\boldsymbol{\Phi}}

\newcommand{\bpsi}{\boldsymbol{\psi}}

\begin{document}

\title
{Non-commutative lattice modified Gel'fand--Dikii systems}

\author{Adam Doliwa}

\address{A. Doliwa, Faculty of Mathematics and Computer Science, University of Warmia and Mazury,
ul.~S{\l}oneczna~54, 10-710~Olsztyn, Poland}

\email{doliwa@matman.uwm.edu.pl}
\urladdr{http://wmii.uwm.edu.pl/~doliwa/}

\date{}
\keywords{integrable difference equations; Gel'fand-Dikii hierarchies; multidimensional consistency; lattice modified Boussinesq equation; Desargues maps; non-commutative Hirota equation; non-isospectral integrable systems}
\subjclass[2010]{37K10, 37K25, 37K60, 39A14, 51P05}

\begin{abstract}
We introduce integrable multicomponent non-commutative lattice systems, which can be considered as analogs of the modified Gel'fand--Dikii hierarchy. We present the corresponding systems of Lax pairs and we show directly multidimensional consistency of these Gel'fand--Dikii type equations. We demonstrate how the systems can be obtained as periodic reductions of the non-commutative lattice Kadomtsev--Petviashvilii hierarchy. The geometric description of the hierarchy in terms of Desargues maps helps to derive non-isospectral generalization of the non-commutative lattice modified Gel'fand--Dikii systems. We show also how arbitrary functions of single arguments appear naturally in our approach when making commutative reductions, which we illustrate on the non-isospectral non-autonomous versions of the lattice modified Korteweg--de~Vries and Boussinesq systems. 
\end{abstract}
\maketitle

\section{Introduction}
The so called Gel'fand--Dikii (GD) hierarchy \cite{GD,Manin} generalizes the 
Korteweg--de~Vries (KdV) hierarchy~\cite{DKJM,SegalWilson} to higher order spectral problems. The lattice GD hierarchy, obtained via the direct linearization approach and given in~\cite{FWN-GD}, is a multicomponent system which goes to the "continuous" hierarchy under suitable limits. The simplest systems are the lattice potential KdV equation~\cite{FWN-Capel-KdV} and a lattice version of the Boussinesq equation. In~\cite{FWN-GD} it was also indicated how to obtain a "modified" version of the lattice GD hierarchy, and its two simplest members were given explicitly. 

The origins of soliton theory can be traced back to classical developments in theory of submanifolds and their transformations~\cite{Bianchi,Darboux-OS,DarbouxIV,Eisenhart-TS,Tzitzeica}. The analogous geometric interpretation of integrable partial difference systems was developed recently starting from~\cite{BP1,BP2,DCN,MQL,BobenkoSchief,KoSchief2,TQL}. Integrability of discrete (Adler)-Gel'fand-Dikii equations was investigated recently using the so called pentagram maps~\cite{Schwartz-pentagram,OvsienkoSchwartzTabachnikov-CMP,Adler-tangential,GMBeffa-AGD-pentagram}. From author's personal viewpoint the key objects in geometric theory of integrable discrete systems are the notions of Desargues maps~\cite{Dol-Des} (see Section~\ref{sec:Des}) and closely related lattices of planar quadrilaterals~\cite{Sauer,DCN,MQL} (quadrilateral lattices for short), which are difference geometry counterpart of conjugate nets investigated thoroughly by Darboux, Bianchi and their followers. 

In recent studies on discrete integrable systems the property of multidimensional consistency \cite{ABS,FWN-cons} is considered as the main concept of the theory. Roughly speaking, it is the possibility of extending the number of independent variables of a given nonlinear system by adding its copies in different directions without creating this way inconsistency or multivaluedness. The multidimensional consistency is considered as "the precise analogue of the \emph{hierarchy} of nonlinear evolution equations in the case of continuous systems"~\cite{NRGO}. Such approach lies also at the roots of an earlier "principle" of discretization of integrable differential equations via iterated applications of B\"{a}cklund transformations \cite{LeviBenguria}; for its application to systems of a geometric origin see also \cite{NiSchief,GanzhaTsarev,TQL}. Recently, multidimensional consistency of GD hierarchy and related multicomponent systems was studied in~\cite{LobbNijhoff-GD-L,Hietarinta-Boussinesq,Atkinson-Lobb-Nijhoff-GD}.

Hirota's discrete Kadomtsev--Petviashvilii (KP) equation \cite{Hirota} plays dominant role in the whole theory of integrable systems both classical and quantum~\cite{KNS-rev}. Its direct geometric counterpart are Desargues maps, while multidimensional quadrilateral lattices give rise to the so called discrete Darboux equations~\cite{BoKo}. Multidimensional consistency of the geometric construction of quadrilateral lattices and of the discrete Darboux equations was discussed already in \cite{MQL}.  Subsequently, certain geometric restrictions compatible with such multidimensionally consistent construction scheme were used~\cite{CDS,q-red,DS-sym,BQL,CQL} to isolate integrable reductions of the quadrilateral lattices.
The important relation of four dimensional consistency of the Hirota equation (in its commutative Schwarzian form) and the Desargues configuration has been observed by Wolfgang Schief (the author has learned about that relation from the talk of Alexander Bobenko~\cite{Bobenko-talk}). Multidimensional consistency of both systems survives the transfer to their non-commutative versions~\cite{Dol-GAQL,Dol-Des}, where as the ambient space of lattice submanifolds one takes the projective space over an arbitrary division ring~\cite{BeukenhoutCameron-H}. One of motivations to study non-commutative versions of integrable discrete systems \cite{FWN-Capel,Kupershmidt,BobSur-nc,Nimmo-NCKP} is their relevance in integrable lattice field theories. In particular, as shown by \cite{BaMaSe,Sergeev-q3w}, the four dimensional consistency of the geometric construction of multidimensional lattice of planar quadrilaterals~\cite{MQL} is related to Zamolodchikov's tetrahedron equation~\cite{Zamolodchikov}, which is a multidimensional analogue of the quantum Yang--Baxter equation~\cite{Baxter,QISM,YB-Miwa}. Recently, the four dimensional consistency of Desargues maps has been related~\cite{DoliwaSergeev-pentagon,Dol-WCR-Hirota} to certain solutions of the functional pentagon quation~\cite{Zakrzewski} and of its quantum reduction, which is of fundamental importance in the modern analytic theory of quantum groups~\cite{Timmermann}.

The non-autonomous and/or non-isospectral versions of integrable equations contain additional functional freedom which, from physical point of view, allows to study solitary waves in non-uniform media with relaxation effects, and therefore provide more realistic models. The relevant mathematical tools are interesting generalization of standard techniques of integrable systems theory both for differential and difference equations, see for example 
\cite{Calogero-ni,Fuchssteiner,GordoaPickering,LeviRagnisco,LeviRagniscoRodriguez,
MaFuchssteiner,Cieslinski}. It is remarkable that certain non-isospectral equations relevant in gravity theory \cite{Ernst} were known in a purely geometric context of surface theory \cite{Bianchi,LeviSym}, see also \cite{Schief-Calapso,DoliwaNieszporskiSantini-Bianchi-asympt} where corresponding discrete system was investigated. Recently a non-commutative non-isospectral KP equation was studied in \cite{ZhuZhangLi}. 

In making dimensional symmetry reduction \cite{FWN-lB,NRGO} of a given integrable partial difference system, to achieve full generality of the resulting equation the non-autonomous version of the system is indispensable, see for example \cite{GRSWC} for description of this feature in the case of discrete Painlev\'{e} equations. There are various mechanism of deautonomisation, for example using the singularity 
confinement~\cite{GrammaticosRamaniPapageorgiou} or an appropriate reduction condition starting from non-autonomous version of the Hirota system~\cite{WTS}. In geometric approach to discrete systems there are known cases where non-autonomous factors result from reductions of autonomous integrable systems of higher dimensionality, see for example~\cite{DoliwaNieszporskiSantini-Bianchi-asympt,Doliwa-isoth}. 
In literature there are known examples of non-commutative or quantum Painlev\'{e} equations \cite{Nagoya,RR-ncPII,BertolaCafasso-CMP,Kuroki}. It is desirable to apply similar reduction procedure starting from non-commutative or quantum non-autonomous and/or non-isospectral integrable systems. 

The construction of the paper is as follows. We concentrate on the non-commutative lattice modified GD systems and we start in Section~\ref{sec:GD} by introducing such a system. We present its linear problem and we also show directly multidimensional consistency of the system. Moreover, in making commutative reduction we discover that arbitrary functions of single arguments show up here naturally. In the simplest case we recover the non-autonomous lattice modified KdV equation, and then we present the corresponding non-autonomous lattice modified Boussinesq equation. Then in Section~\ref{sec:Des} we recall relevant facts from geometric theory of the non-commutative Hirota system in order to present the lattice non-commutative modified GD system as reduction of the Desargues maps. We also make connection to a non-commutative version of the KP hierarchy and we show an auxiliary result on allowed gauge transformations of the linear problem for non-commutative Hirota system. Such considerations turn out to be useful in 
Section~\ref{sec:per-red} where we consider periodic reduction of Desargues maps and of the non-commutative KP hierarchy. In doing that we derive non-isospectral generalization of the non-commutative lattice modified Gel'fand--Dikii system. We show its integrability by presenting the corresponding Lax system and by proving directly its multidimensional consistency. We also discuss commutative reduction of the system and we derive non-isospectral non-autonomous versions of the lattice modified KdV and Boussinesq systems. The concluding section discusses some related points of current and future research.

\section{Lattice modified Gel'fand--Dikii systems}
\label{sec:GD}
For a function $f$ defined on $N$-dimensional integer lattice $\ZZ^N$ we denote by $f_{(i)}$ its translate in the variable $n_i\mapsto n_i + 1$, $i=1,\dots,N$, i.e. 
$f_{(i)}(n_1,\dots,n_i,\dots ,n_N) = f(n_1,\dots,n_i + 1,\dots ,n_N)$.
Consider the following system of partial difference equations
\begin{align} \label{eq:GD-K}
(r_{k(j)}^{-1} - r_{k(i)}^{-1}) r_{k(ij)} & = r_{k+1}^{-1} (r_{k+1(i)} - r_{k+1(j)} ), \qquad k=1, \dots , K-1,\\ \nonumber
(r_{K(j)}^{-1} - r_{K(i)}^{-1}) r_{K(ij)} & = r_{1}^{-1} (r_{1(i)} - r_{1(j)} ), \qquad \qquad i\neq j , 
\end{align}
where $r_k :\ZZ^N\to\DD$, $k=1, \dots ,K$, are unknown functions from $N\geq 2$ dimensional integer lattice $\ZZ^N\ni (n_1, \dots , n_N) = n$ taking values in a division ring (skew field) $\DD$. To compactify notation, in this Section we consider the index $k$ modulo $K$ (starting from $k=1$). 

Equations \eqref{eq:GD-K} provide compatibility of the linear system
\begin{equation} \label{eq:LmA}
(\bpsi_1, \dots , \bpsi_K)_{(i)} = 
 (\bpsi_1, \dots , \bpsi_K) \left(  \begin{array}{ccccc}  
1 & 0  & \cdots & 0 & \lambda r_1 r_{K(i)}^{-1} \\
\lambda r_2 r_{1(i)}^{-1} & 1 & 0 & \hdots  & 0 \\
0 & \lambda r_3 r_{2(i)}^{-1} & \ddots &  &  \vdots \\
\vdots &  & & 1 & 0 \\
0 & 0 & \ \hdots  & \lambda r_K r_{K-1(i)}^{-1}  & 1\end{array} \right) ,
\end{equation}
where $\lambda$ is a central (spectral) parameter. 
We remark that due to thier geometric meaning (see Section~\ref{sec:Des}) we can treat $\bpsi_k$ as column vectors. 

\subsection{Relation to the lattice modified Korteweg--de~Vries and Boussinesq equations}  \label{sec:mKdV-B}
The linear problem \eqref{eq:LmA} in the simplest case $K=2$ and the corresponding nonlinear system (in the commutative case) was studied in \cite{Hay-JMP} as the double lattice modified KdV system. In this Section we show how to approach to this system from general perspective. When the division ring $\DD$ is commutative then there naturally  show up functions of single variables, which will play the role of non-autonomous factors.
\begin{Lem} \label{lem:prod}
Given solution $r_k$, $k=1,\dots ,K$ of the GD system  \eqref{eq:GD-K}, if the division ring $\DD$ is commutative then the function $R=r_1 r_2 \dots r_K$, splits into the product of functions of single variables $n_i$, $i=1,\dots , N$.
\end{Lem}
\begin{proof}
Multiply all the $K$ equations and use the commutativity assumption to obtain that in the generic case (i.e. when the differences $r_{k(i)} - r_{k(j)}$ do not vanish for $i\neq j$)
\begin{equation} \label{eq:RR}
R_{(ij)} R = R_{(i)} R_{(j)},  \qquad i\neq j,
\end{equation}
which implies the statement.
\end{proof}
Let us consider in more detail the simplest case $K=2$. Denote by $G(n) = G_1(n_1) \dots G_N(n_N)$ the product of functions of single variables such that the function $R$ in Lemma~\ref{lem:prod} reads $R=G^2$ (we assume that square roots exist --- we take such a form for convenience to match known results and to make resulting equations more symmetric). Moreover, define the functions $F_i$ of single variable $n_i$ as "discrete logarithmic derivatives" of $G_i$, i.e. $G_i(n_i + 1) = F_i(n_i) G_i(n_i)$, $i=1,\dots , N$. We parametrize the unknown functions $r_1$ and $r_2$ in terms of a single function $x$ as follows 
\begin{equation}
r_1 = xG, \qquad r_2 = \frac{G}{x}.
\end{equation}
Then the system \eqref{eq:GD-K} takes the form of the well known non-autonomous lattice modified KdV system \cite{Hirota-KdV,Hirota-sG,FWN-Capel-KdV,PapageorgiouGrammaticosRamani}
\begin{equation} \label{eq:l-mKdV} 
x_{(ij)} = x \;\frac{x_{(i)}F_j - x_{(j)}F_i}{x_{(j)}F_j - x_{(i)}F_i} , \qquad i\neq j.
\end{equation}
In such a form the functions $F_i$ are arbitrary and depend on the model under consideration. Therefore the $K=2$ field version of equations \eqref{eq:GD-K} can be called the non-commutative (and in a sense also non-autonomous) lattice modified KdV system. Actually, for $N> 2$ it is fair to call it the non-commutative lattice modified KdV hierarchy; see also Section~\ref{sec:consistency}.
\begin{Rem}
In \cite{BobSur-nc} a different (one-component and with central non-autonomous factors/parameters) non-commutative version of the lattice modified KdV system \eqref{eq:l-mKdV} was studied.
\end{Rem}

Consider next the case $K=3$, assuming still commutativity of $\DD$, and setting this time  $R=G^3$ keeping the definition of $F_i$ unchanged.  Then we parametrize the unknown functions $r_1$, $r_2$ and $r_3$ in terms of two functions $x$ and $y$ as follows
\begin{equation*}
r_1 = xG, \qquad r_2 = \frac{yG}{x}, \qquad r_3= \frac{G}{y},
\end{equation*}
which leads to a new system (the non-autonomous version of that studied in \cite{FWN-lB,Atkinson-Lobb-Nijhoff-GD})
\begin{equation}
\label{eq:l-B-2} 
x_{(ij)} = \frac{x}{y} \;
\frac{y_{(j)} x_{(i)}F_j - y_{(i)} x_{(j)}F_i}{x_{(j)}F_j - x_{(i)}F_i},
\qquad y_{(ij)} = x \; \frac{y_{(i)}F_j - y_{(j)}F_i}{x_{(j)}F_j - x_{(i)}F_i}, \qquad 
i\neq j, 
\end{equation}
which could be rewritten as the non-autonomous lattice modified Boussinesq equation 
\begin{equation} \label{eq:na-l-m-B}
\begin{split}
\left( \frac{ y_{(i)} F_j - y_{(j)} F_i }{ y_{(ij)} } \right)_{(ij)} 
- y \left( \frac{ F_j }{ y_{(j)} } -  \frac{ F_i }{ y_{(i)} } \right) & =  \\
= \left(  \frac{ y_{(ij)} }{ y } \, 
\frac{ y_{(j)} F_j^2 -y_{(i)} F_i^2 }{ y_{(i)} F_j - y_{(j)} F_i } \right)_{(j)}  - &
\left(  \frac{ y_{(ij)} }{ y } \, 
\frac{ y_{(i)} F_i^2 -y_{(j)} F_j^2 }{ y_{(j)} F_i - y_{(i)} F_j } \right)_{(i)} , \qquad i\neq j,
\end{split}
\end{equation}
which is a direct generalization (functions of single variables replace parameters) of the equation considered in \cite{FWN-lB}.

Similar procedure can be applied, in principle, for arbitrary $K$. It is convenient to start from the function $R$ of the form $R=G^K$ (we assume existence of roots of appropriate orders, eventually we can consider the field extensions) and parametrize functions $r_k$ as follows
\begin{equation*}
r_1 = \tau_1 G, \qquad r_2 = \frac{\tau_2}{\tau_1}G, \quad \dots , \qquad r_{K-1} = \frac{\tau_{K-1}}{\tau_{K-2}}G, \qquad r_K = \frac{G}{\tau_{K-1}}.
\end{equation*}
If we put $\tau_0=1$ and  $\tau_{k+K} = \tau_k$ then in the autonomous case we recover the approach to the lattice modified GD systems described in \cite{Atkinson-Lobb-Nijhoff-GD} as reductions of the Hirota 
equation~\cite{Hirota}. Our case can be recovered from the non-autonomous version of the Hirota equation~\cite{WTS} for appropriate choice of the non-autonomous factors; see Section~\ref{sec:per-red}. We remark that the Lax matrices \eqref{eq:LmA} differ from those considered in \cite{Atkinson-Lobb-Nijhoff-GD}. Our linear system is chosen to match the Lax pair for the double lattice modified KdV system \cite{Hay-JMP}. There exists simple transition between the two Lax systems which will be given in Section~\ref{sec:Des}.

\subsection{Multidimensional consistency of the system}
\label{sec:consistency}

Equations \eqref{eq:GD-K} form closed systems for $N=2$ of independent variables. However writing it for arbitrary $N\geq 2$ we implicitly considered it as multidimensionally consistent, what we are going to demonstrate. In Section \ref{sec:Des} we show that equations  \eqref{eq:GD-K} can be obtained as a periodic reduction of the lattice non-commutative KP system, which is itself multidimensionally consistent \cite{Dol-Des}. From that point of view the following result can be considered as the proof of integrability of the periodic reduction. 
\begin{Th} \label{th:3D-cons}
The non-commutative lattice modified GD system \eqref{eq:GD-K} is three-dimensionally consistent.
\end{Th}
\begin{proof}
To show that the two expressions below are equal (below we assume that the indices $i,j,l$ are distinct)
\begin{align} \label{eq:M-R1}
[r_{k(il)}]_{(j)} & = (r_{k(lj)}^{-1} - r_{k(i j)}^{-1} )^{-1} r_{k+1(j)}^{-1} 
(r_{k+1(ij)} - r_{k+1 (l j)}) ,\\ \label{eq:M-R2}
[r_{k(ij)}]_{(l)} & = (r_{k(jl)}^{-1} - r_{k(i l)}^{-1} )^{-1} r_{k+1(l)}^{-1} 
(r_{k+1(i l)} - r_{k+1 (j l)}) ,
\end{align}
it is convenient to note first the following identity
\begin{equation} \label{eq:LR-identity}
(r_{k(i)}^{-1} - r_{k(j)}^{-1} ) r_{k(ij)} (r_{k(ij)}^{-1} - r_{k(l j)}^{-1} ) =
(r_{k(i)}^{-1} - r_{k(l)}^{-1} ) r_{k(il)} (r_{k(il)}^{-1} - r_{k(j l)}^{-1} ),
\end{equation}
which can be verified directly by simple algebraic manipulation taking into account equations \eqref{eq:GD-K}. Notice that generically the identity \eqref{eq:LR-identity} is non-trivial meaning that its sides do not vanish. Then we consider the expression 
$E= L R_1 - R R_2$, where $L/R$ denotes left/right hand side of identity 
\eqref{eq:LR-identity}, and $R_1/R_2$ is the right hand side of equation \eqref{eq:M-R1}/\eqref{eq:M-R2}. We show below the main steps of the calculation
\begin{align*}
E & = r_{k+1}^{-1}\left[ 
(r_{k+1(i)} - r_{k+1(j)} ) r_{k+1(j)}^{-1}
(r_{k+1(ij)} - r_{k+1(lj)} ) - (r_{k+1(i)} - r_{k+1(l)} ) r_{k+1(l)}^{-1} 
(r_{k+1(il)} - r_{k+1(jl)} ) \right] \\ &= r_{k+1}^{-1} r_{k+1(i)} \left[
(r_{k+1(j)}^{-1} - r_{k+1(i)}^{-1} ) r_{k+1(ij)} - (r_{k+1(l)}^{-1} - r_{k+1(i)}^{-1} ) r_{k+1(il)} - r_{k+2}^{-1}(r_{k+2(l)} - r_{k+2(j)} ) \right] =0,
\end{align*}
which concludes the proof. We only remark that periodicity in the index $k$ was used in the proof but it was not so crucial like in the proof of Lemma~\ref{lem:prod}.
\end{proof}

\section{Desargues maps, non-commutative Hirota system, and the lattice KP hierarchy}
\label{sec:Des}
In this Section we recapitulate first main elements of the geometric approach to non-commutative Hirota system \cite{Nimmo-NCKP} via the so-called Desargues maps \cite{Dol-Des}; see also \cite{Schief-talk,KoSchief-Men} to compare with earlier related works. We also derive other properties of the equations which will be useful in performing periodic reduction to (non-isospectral and non-commutative) lattice modified GD system. 
\subsection{Desargues maps and the Hirota system}
Consider Desargues maps 
$\Phi\colon\ZZ^{\tilde{N}}\to\PP^M(\DD)$, which are characterized by the condition that for arbitrary $n\in\ZZ^{\tilde{N}}$ the points
$\Phi(n)$, $\Phi_{(i)}(n)$ and  $\Phi_{(j)}(n)$, for $i\neq j$, are collinear (see Fig.~\ref{fig:Desargues-Veblen}).
\begin{figure}
\begin{center}
\includegraphics[width=7cm]{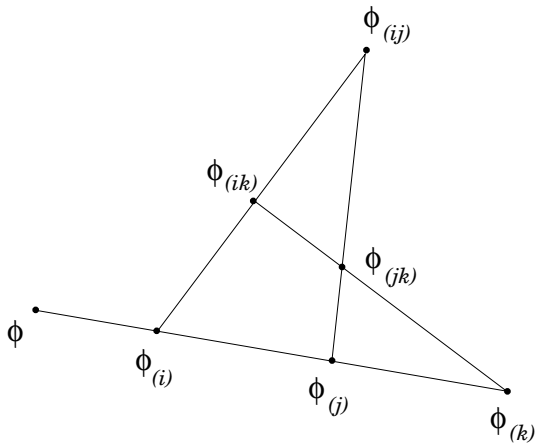}
\end{center}
\caption{Desargues map condition and the Veblen configuration.} 
\label{fig:Desargues-Veblen}
\end{figure}
In the homogeneous coordinates $\bPhi\colon\ZZ^{\tilde{N}}\to\DD^{M+1}$ the defining condition of Desargues maps can be described in terms of a linear relation between $\bPhi$, $\bPhi_{(i)}$ and $\bPhi_{(j)}$. As it was shown in \cite{Dol-Des} there exists a particular normalization of the homogeneous coordinates (gauge) in which
the linear relation takes the form
\cite{DJM-II,Nimmo-NCKP} (we consider right vector spaces over division rings)
\begin{equation} \label{eq:lin-dKP}
\bPhi_{(i)} - \bPhi_{(j)} =  \bPhi U_{ij},  \qquad i \ne j \leq \tilde{N},
\end{equation}
and the functions $U_{ij}\colon\ZZ^{\tilde{N}}\to\DD$ satisfy
to the non-commutative Hirota (or the non-Abelian Hirota--Miwa \cite{Nimmo-NCKP}) system.
\begin{align} \label{eq:alg-comp-U}
& U_{ij} + U_{jl} + U_{li} = 0, \qquad U_{ij}+U_{ji} = 0,\\
& \label{eq:U-rho} 
U_{lj}U_{li(j)} = U_{li} U_{lj(i)},
\end{align}
for distinct indices $i,j,l$.
Equation \eqref{eq:U-rho} allows to introduce 
the potentials $\rho_i\colon\ZZ^{\tilde{N}}\to\DD^\times$ such that
\begin{equation} \label{eq:def-rho}
U_{ij} = \rho_i^{-1}
 \rho_{i(j)}.
\end{equation}

One may ask how big is the class of gauge  functions $F \colon\ZZ^{\tilde{N}}\to\DD^\times$ such that after the transformation $\bPhi \to \hat{\bPhi} = \bPhi F^{-1}$ the form \eqref{eq:lin-dKP} of the linear system remains unchanged. The following results can be demonstrated by direct calculation.
\begin{Prop} \label{prop:gauge-Hir}
The gauge functions of the allowed class are characterized by the condition $F_{(i)} = F_{(j)}$ for all pairs of indices, which means that $F$ is a function of the sum $n_{\tilde{\sigma}}=n_1 + n_2 + \dots + n_{\tilde{N}}$ of the discrete variables. Denote $F_{(\tilde{\sigma})}(n_{\tilde{\sigma}}) = F(n_{\tilde{\sigma}} +1)$ then  the transformed vector-function $\hat{\bPhi}$ satisfies equation 
\begin{equation} \label{eq:allowed-gauge}
\hat{\bPhi}_{(i)} - \hat{\bPhi}_{(j)} =  \hat{\bPhi} \hat{U}_{ij},  \qquad i \ne j \leq \tilde{N}, \qquad \text{with} \quad \hat{U}_{ij} = F U_{ij} F_{(\tilde{\sigma})}^{-1}.
\end{equation}
\end{Prop}
\begin{Cor}
The resulting transformation $U_{ij}\to\hat{U}_{ij} $ provides a symmetry of the non-commutative Hirota system \eqref{eq:alg-comp-U}-\eqref{eq:U-rho}. The corresponding transformation of the potentials reads $\hat{\rho}_i = \rho_i F^{-1}$. 
\end{Cor}

We remark that when $\DD$ is commutative then the simpler part of the algebraic relations \eqref{eq:alg-comp-U} implies that the functions $\rho_i$ can 
be parametrized in terms of a 
single potential $\tau$ (the tau-function)
\begin{equation} \label{eq:r-tau}
\rho_i = (-1)^{\sum_{j>i}n_j}
\frac{\tau_{(i)}}{\tau}
\end{equation} 
Then remaining equations
\eqref{eq:alg-comp-U} reduce to the celebrated Hirota system \cite{Hirota}
\begin{equation} \label{eq:H-M}
\tau_{(i)}\tau_{(jl)} - \tau_{(j)}\tau_{(il)} + \tau_{(l)}\tau_{(ij)} =0,
\qquad 1\leq i< j < l \leq \tilde{N}.
\end{equation}
\begin{Cor}
Having still $\DD$ commutative, denote by $\mathcal{M}$ the "discrete logarithmic integral" of the function $F$, i.e. $\mathcal{M}_{(\tilde{\sigma})} = F \mathcal{M}$. Then during transformation \eqref{eq:allowed-gauge} the $\tau$-function changes to $\hat{\tau} = \tau/ \mathcal{M}$, which provides symmetry of the Hirota system \eqref{eq:H-M}. 
\end{Cor}

In Section~\ref{sec:per-red} we will need the following non-autonomous version \cite{WTS} of the Hirota system \eqref{eq:H-M} 
\begin{equation} \label{eq:H-M-na}
(A_j - A_l) \tau_{(i)}\tau_{(jl)} + (A_l - A_i) \tau_{(j)}\tau_{(il)} + (A_i - A_j) 
\tau_{(l)}\tau_{(ij)} =0,
\qquad 1\leq i< j < l \leq \tilde{N},
\end{equation}
where $A_i$ is an arbitrary function of the variable $n_i$, $i=1, \dots ,\tilde{N}$. As it was discussed in \cite{WTS} both equations \eqref{eq:H-M} and \eqref{eq:H-M-na} are equivalent when considered prior to imposing reductions. 

\begin{Rem}
We remark that the "Hirota equation" studied in \cite{Faddeev-Volkov-Hirota,BobSur-nc} is the (autonomous version of the) lattice modified KdV
equation~\eqref{eq:l-mKdV}, see also \cite{Hirota-KdV,Hirota-sG}. Equation \eqref{eq:H-M} in the basic case of $\tilde{N}=3$ discrete variables was originally described as "discrete analogue of a generalized Toda equation". It is also called "Hirota--Miwa" system after Miwa showed~\cite{Miwa} its fundamental role in the whole KP hierarchy.
\end{Rem}

\subsection{Non-commutative lattice Kadomtsev-Petviashvilii systems}
Let $\tilde{N} = N+1$, we distinguish the last variable $k=n_{N+1}$. To match notation of \cite{KNY-qKP} we denote also $n=(n_1, \dots , n_N)$, and
\begin{equation*}
\bPhi(n,k) = \bphi_k(n), \quad U_{N+1,i}(n,k) = u_{i,k}(n), \quad \rho_{N+1}(n,k) = r_k(n), \quad \tau(n,k) = \tau_k(n), \quad F(n,k) = f_k(n),
\end{equation*}
which allows the rewrite a part (that with the distinguished variable) of the linear problem \eqref{eq:lin-dKP} in the form
\begin{equation} \label{eq:lin-lKPh}
\bphi_{k+1} - \bphi_{k(i)} = \bphi_k u_{i,k}, \qquad i=1,\dots ,N.
\end{equation}
The compatibility of the above linear system reads
\begin{align} \label{eq:ncKP-1}
& u_{j,k}u_{i,k(j)} = u_{i,k} u_{j,k(i)}, \qquad i\neq j,\\
& u_{i,k(j)} + u_{j,k+1} = u_{j,k(i)} + u_{i,k+1}. \label{eq:ncKP-2}
\end{align}
Equations \eqref{eq:ncKP-1} allow (what we knew already) to define potentials $r_{k}$ such that $u_{i,k} = r_{k}^{-1} r_{k(i)}$, while equations \eqref{eq:ncKP-2} give the system 
\begin{align} \label{eq:KP}
(r_{k(j)}^{-1} - r_{k(i)}^{-1}) r_{k(ij)} & = r_{k+1}^{-1} (r_{k+1(i)} - r_{k+1(j)} ), \qquad i\neq j,
\end{align}
Following terminology of \cite{KNY-qKP} we can call the system \eqref{eq:ncKP-1}-\eqref{eq:ncKP-2} (or equations \eqref{eq:KP}) the non-commutative lattice KP hierarchy. In fact, from \eqref{eq:lin-lKPh} and \eqref{eq:ncKP-1}-\eqref{eq:ncKP-2} we can recover all other ingredients of the non-commutative Hirota system 
\begin{align*}
U_{i, N+1}(n,k) & = - u_{i,k}(n) = -  r_{k}(n)^{-1} r_{k(i)}(n), \\ 
U_{ij}(n,k) = - U_{j,N+1}(n,k) - U_{N+1,i}(n,k) &= u_{j,k}(n) - u_{i,k}(n) =  r_{k}(n)^{-1} ( r_{k(j)}(n) - r_{k(i)}(n) ).
\end{align*}
One can check that the remaining equations of the system \eqref{eq:alg-comp-U}-\eqref{eq:U-rho} are either trivially satisfied or are equivalent to \eqref{eq:ncKP-1}-\eqref{eq:ncKP-2} (or \eqref{eq:KP} without periodicity assumption). We remark that identity \eqref{eq:LR-identity} which turned out to be helpful in the proof of Theorem~\ref{th:3D-cons} was suggested by equation 
$U_{ij} U_{il(j)}= U_{il} U_{ij(l)}$. Notice that the proof of Theorem~\ref{th:3D-cons} goes over to the "infinite-component" case \eqref{eq:KP}.

Let us present the following consequence of Proposition~\ref{prop:gauge-Hir}.
\begin{Cor} \label{cor:gauge-KP}
A system of gauge functions $f_k\colon\ZZ^N\to\DD^\times$, $k\in\ZZ$, leaving the form of the linear system  \eqref{eq:lin-lKPh} unchanged is characterized by the condition $f_{k(i)} = f_{k+1}$.  Then the transformation
\begin{equation}
\bphi_k \to \hat{\bphi}_k = \bphi_k f_k^{-1}, \qquad
u_{i,k}\to \hat{u}_{i,k} = f_k u_{i,k} f_{k+1}^{-1}, \qquad 
r_k \to \hat{r}_k = r_k f_k^{-1}, 
\end{equation}
provides symmetry of equations \eqref{eq:ncKP-1}-\eqref{eq:ncKP-2} and of the system
\eqref{eq:KP}, correspondingly.
\end{Cor}

\section{Periodic reductions of Desargues maps}
\label{sec:per-red}
\subsection{Non-commutative lattice modified Gel'fand--Dikii systems}
At this point we can impose the periodic reduction $\Phi(n,k+K) = \Phi(n,K)$ on the level of Desargues maps, and study its implications on the level of the functions 
$\bphi_k$ and $u_{i,k}$. 
\begin{Rem}
Without entering into details we mention that interpretation~\cite{Dol-Des} of Desargues maps in terms of quadrilateral lattices~\cite{MQL} and their Laplace transformations~\cite{DCN,TQL} provides geometric indication of integrability of the periodic reduction. This follows from the observation that the Laplace transforms of a periodic quadrilateral lattice preserve the periodicity condition.
\end{Rem}
Then, by definition of the homogeneous coordinates, there exist functions (the monodromy factors) $\mu_k\colon \ZZ^N\to\DD^\times$ such that, 
$\bphi_{k+K} = \bphi_k \mu_k$. 
\begin{Prop}
The monodromy factors satisfy the condition $\mu_{k+1} = \mu_{k(i)}$, $i=1,\dots ,N$.
\end{Prop}
\begin{proof}
The monodromy factors do not change the structure of the linear problem \eqref{eq:lin-lKPh}, therefore the conclusion follows from Corollary~\ref{cor:gauge-KP}.
\end{proof}
\begin{Cor}
The corresponding transformation of the potentials $u_{i,k}$ and $r_k$ is given by
\begin{equation}
u_{i,k+K} = \mu_k^{-1} u_{i,k}\mu_{k(i)}, \qquad 
r_{k+K} = r_k \mu_k.
\end{equation}
\end{Cor}
The periodicity condition on the geometric level implies then the following linear system
\begin{equation} \label{eq:Lm-kp-K}
\left(\bphi_1, \bphi_2, \dots , \bphi_K \right)_{(i)}= 
\left(\bphi_1, \bphi_2, \dots , \bphi_K \right)
\left(  \begin{array}{ccccc}  
-u_{i,1} & 0  & \cdots & 0 & \mu_1 \\
1 & -u_{i,2}  & 0 & \hdots  & 0 \\
0 & 1 & \ddots &  &  \vdots \\
\vdots &  & & -u_{i,K-1}  & 0 \\
0 & 0 & \ \hdots  & 1 & -u_{i,K}  \end{array} \right) ,
\end{equation}
where $\mu_1$ is a function of the variable $n_\sigma = n_1 + \dots + n_N$. The corresponding system of non-linear equations takes the form
\begin{align} 
\label{eq:GD-K-mu}
(r_{k(j)}^{-1} - r_{k(i)}^{-1}) r_{k(ij)} & = r_{k+1}^{-1} (r_{k+1(i)} - r_{k+1(j)} ), \quad k=1,\dots ,K-1, \\ \nonumber
(r_{K(j)}^{-1} - r_{K(i)}^{-1}) r_{K(ij)} & = 
\mu_1^{-1} r_{1}^{-1} (r_{1(i)} - r_{1(j)} ) \mu_{1(\sigma)} \qquad i\neq j.
\end{align}
When $\mu_1$ is a constant from the center of $\DD$ then the nonlinear system 
\eqref{eq:GD-K-mu} takes the form \eqref{eq:GD-K}. In such a case we may introduce for convenience new central parameter $\lambda$ by $\mu_1 = (-\lambda)^K$ and define
\begin{equation}
\bpsi_k = (-\lambda)^{-k} (-1)^{n_\sigma} \bphi_k r_k^{-1}, \qquad k=1, \dots , K,
\end{equation} 
then the linear problem \eqref{eq:Lm-kp-K} is replaced by \eqref{eq:LmA}.

\subsection{Non-isospectral lattice modified Gel'fand--Dikii systems}
The system \eqref{eq:GD-K-mu} can be considered as a "non-isospectral" version of the non-commutative lattice modified GD hierarchy \eqref{eq:GD-K}. 
In this Section we concentrate on properties of the system  \eqref{eq:GD-K-mu} and of its commutative reduction. 
\begin{Prop} \label{prop:GD-ni-3D}
The non-commutative non-isospectral lattice modified GD system is three-dimensionally consistent.
\end{Prop}
\begin{proof}
In view of the proof of Theorem~\ref{th:3D-cons} it is enough to examine the consistency for $k=K-1, K$ checking before that identity \eqref{eq:LR-identity} holds for $k=K$. 
\end{proof}
As we have already mentioned, the multicomponent reductions of the non-commutative Hirota system we consider are meaningful in the case of $N=2$ independent variables. Other independent variables in the countable number can be considered as symmetries of the original equation. The same applies to the non-isospectral monodromy factor where the sum of remaining variables gives an additive parameter within the argument of the $\mu_1$ function.

Let us examine the non-isospectral analogue of Lemma~\ref{lem:prod}.
\begin{Lem} \label{lem:prod-ni}
Given solution $r_k$, $k=1,\dots ,K$ of the non-isospectral GD system  \eqref{eq:GD-K-mu}, if the division ring $\DD$ is commutative then the function $R=r_1 r_2 \dots r_K$, splits into the product of functions of single variables $n_i$, $i=1,\dots , N$, and a function of the sum of independent variables.
\end{Lem}
\begin{proof}
Multiplying all the $K$ equations of  \eqref{eq:GD-K-mu} and using the commutativity assumption we obtain that in the generic case 
\begin{equation}
R_{(ij)} R = R_{(i)} R_{(j)} \frac{\mu_{1(\sigma)}}{\mu_1},  \qquad i\neq j.
\end{equation}
Introduce the function $m_1$ of the variable $n_{\sigma}$ as the "discrete logarithmic integral" of the monodromy factor: $m_{1(\sigma)} = \mu_1 m_1$. Then the function $\hat{R} = R m_1^{-1}$ satisfies equation \eqref{eq:RR}, which gives the rest of the factorization.
\end{proof}

To find the integrable non-isospectral modification of the non-autonomous lattice modified KdV system \eqref{eq:l-mKdV} we repeate considerations of Section~\ref{sec:mKdV-B} for $K=2$ with 
\begin{equation*}
r_1 = xG, \qquad r_2 = \frac{G m_1}{x} ,
\end{equation*}
then equations \eqref{eq:GD-K-mu} reduce to  
\begin{equation} \label{eq:l-mKdV-ni} 
x_{(ij)} = x \mu_1 \;\frac{x_{(i)}F_j - x_{(j)}F_i}{x_{(j)}F_j - x_{(i)}F_i} , \qquad i\neq j.
\end{equation}

In the case $K=3$ of the non-isospectral lattice modified Boussinesq equation we parametrize the unknown functions $r_1$, $r_2$ and $r_3$ in terms of two functions $x$ and $y$ as follows
\begin{equation*}
r_1 = xG, \qquad r_2 = \frac{yG}{x}, \qquad r_3= \frac{G m_1}{y},
\end{equation*}
which leads to the system 
\begin{equation}
\label{eq:l-B-2-ni} 
x_{(ij)} = \frac{x }{y } \;
\frac{y_{(j)} x_{(i)}F_j - y_{(i)} x_{(j)}F_i}{x_{(j)}F_j - x_{(i)}F_i},
\qquad y_{(ij)} =\mu_1 x \; 
\frac{y_{(i)}F_j - y_{(j)}F_i}{x_{(j)}F_j - x_{(i)}F_i}, \qquad 
i\neq j. 
\end{equation}
It can be rewritten as a scalar equation 
\begin{equation} \label{eq:na-l-m-B-ni}
\begin{split}
\left( \frac{ \mu_1 }{  y_{(ij)} }  
\left( y_{(i)} F_j - y_{(j)} F_i \right) \right)_{(ij)} 
-  & \; \mu_1 y \;
\left( \frac{ F_j }{ y_{(j)} } -  \frac{ F_i }{ y_{(i)} } \right)  =  \\
= 
\left(  \frac{ y_{(ij)} }{ \mu_1 y } \, 
\frac{ y_{(j)} F_j^2 -y_{(i)} F_i^2 }{ y_{(i)} F_j - y_{(j)} F_i } \right)_{(j)}  & - 
\left(  \frac{ y_{(ij)} }{ \mu_1 y } \, 
\frac{ y_{(i)} F_i^2 -y_{(j)} F_j^2 }{ y_{(j)} F_i - y_{(i)} F_j } \right)_{(i)}  , \qquad i\neq j.
\end{split}
\end{equation}

We finally remark, that in the commutative case the non-autonomous Hirota system 
\eqref{eq:H-M-na} with the distinguished last variable and the specification
\begin{equation*}
A_i = F_i + 1, \qquad 1\leq i \leq N, \qquad A_{N+1} =1,
\end{equation*}
gives $\tau$-function formulation of the (non-autonomous and commutative) KP hierarchy
\begin{equation}
F_j\tau_{k(i)} \tau_{k+1 (j)} - F_i\tau_{k(j)} \tau_{k+1 (i)} + (F_i - F_j) \tau_{k+1} \tau_{k(ij)} = 0, \qquad i\neq j.
\end{equation}
Its consequence
\begin{equation}
\left( \left( \frac{\tau_{k}}{\tau_{k+1}}\right)_{(j)} F_i - \left( \frac{\tau_{k}}{\tau_{k+1}}\right)_{(i)} F_j \right)  \left( \frac{\tau_{k+1}}{\tau_{k}}\right)_{(ij)} =
\frac{\tau_{k+1}}{\tau_{k+2}}
\left( \left( \frac{\tau_{k+2}}{\tau_{k+1}}\right)_{(i)} F_i - 
\left( \frac{\tau_{k+2}}{\tau_{k+1}}\right)_{(j)} F_j \right) ,
\end{equation}
after identification 
\begin{equation*}
r_k = \frac{\tau_{k+1}}{\tau_{k}}G, \qquad G_{(i)} = F_i G,
\end{equation*}
gives equations \eqref{eq:KP}. By imposing periodicty condition 
\begin{equation}
\tau_{k+K} = \tau_k m_k, \qquad m_{k+1} = m_{k(i)} = \mu_k m_k,
\end{equation} 
we obtain (commutative version of) equations~\eqref{eq:GD-K-mu},  while trivial monodromy gives equations~\eqref{eq:GD-K}. 

\section{Conclusions}
In this article we introduced a non-commutative and integrable (in the sense of multidimensional consistency) version of the lattice modified Gel'fand-Dikii equations as periodic reductions of the non-commutative Hirota system. In the commutative specification our approach gives rise to presence of single argument functions in the lattice GD equations. The simplest case of period two leads to the well known non-autonomous lattice modified Korteweg--de~Vries equation, but the period three case gives a new non-autonomous version of the lattice modified Boussinesq equation. We mention that such type of deautonomization, where parametres present in the integrable equation are replaced by functions of single variables, can be encountered in previous works~\cite{ABS} where it was related to locality of the multidimensional consistency principle. We presented yet another 
(non-isospectral)  generalization of the lattice GD equations which involves an additional function of single variable. The origin of that generalization is related to geometric meaning of the Hirota system in terms of the Desargues maps. A more standard way of deriving the (commutative but non-isospectral and non-autonomous) lattice GD equations from the $\tau$-function of the non-autonomous Hirota system has been presented as well.

In recent works \cite{DoliwaSergeev-pentagon,Dol-WCR-Hirota} we investigated the transition from Desargues maps and the non-commutative Hirota system to the corresponding quantum case, which allowed to rederive algebraic properties of the quantum plane. It would be interesting to study from that point of view the non-commutative lattice GD equations and to find corresponding reduction of the relevant algebraic structures. The present paper is a part of a general program of deriving integrable systems from Desargues maps (or from the Hirota system). In particular it would be desirable to exploit free functions of single arguments in the non-isospectral and non-autonomous lattice modified GD equations in order to obtain in this way~\cite{GRSWC,KNY-qKP,Ormerod} such distinguished systems like the $Q4$ equation \cite{Q4} or more advanced lattice Painleve equations~\cite{Sakai}.

\section*{Acknowledgments}
The research was supported in part by Polish Ministry of Science and Higher Education grant No.~N~N202~174739.

\bibliographystyle{amsplain}

\providecommand{\bysame}{\leavevmode\hbox to3em{\hrulefill}\thinspace}

\end{document}